\pdfoutput=1
\documentclass[conference]{IEEEtran}

\usepackage[utf8]{inputenc}
\usepackage[T1]{fontenc}
\usepackage[x11names,svgnames,usenames,dvipsnames,table,rgb]{xcolor}
\usepackage[english]{babel}
\usepackage{csquotes}
\usepackage{graphicx}
\usepackage[binary-units=true, detect-all=true]{siunitx}
\usepackage{yquant}
\yquantset{register/default name=$\ket{\reg_{\idx}}$}
\usepackage{mathtools,amssymb,amsthm,amsfonts,nicefrac}
\usepackage{physics}
\usepackage{comment}
\usepackage{booktabs}
\usepackage[
	backend=biber,
	style=ieee,
    dashed=false,
	sortcites=true,
	url=true,
	eprint=true,
	giveninits=false,
    uniquename=true,
	minnames=1,
	maxnames=3
	]{biblatex}
\usepackage{csvsimple}
\makeatletter
\let\MYcaption\@makecaption
\makeatother
\usepackage{subcaption}
\makeatletter
\let\@makecaption\MYcaption
\makeatother

\usepackage[textsize=tiny]{todonotes}

\usepackage[hidelinks]{hyperref}

\hypersetup{
    unicode=false,          %
    pdftoolbar=false,        %
    pdfmenubar=true,        %
    pdffitwindow=false,     %
    pdfstartview={FitH},    %
    pdftitle={},    %
    pdfauthor={Tom Peham, Nina Brandl, Richard Kueng, Robert Wille, Lukas Burgholzer},     %
    pdfsubject={SAT Encoding for Optimal Clifford Circuit Synthesis},   %
    pdfcreator={},   %
    pdfproducer={}, %
    pdfkeywords={circuit synthesis, Clifford group, SAT}, %
    pdfnewwindow=true,      %
    colorlinks=true,       %
    linkcolor=red,          %
    citecolor=blue,        %
    filecolor=teal,      %
    urlcolor=blue           %
}

\AtEveryBibitem{%
	\clearname{editor}%
	\clearfield{series}%
	\clearfield{isbn}%
	\clearfield{issn}%
	\clearfield{day}%
	\clearfield{month}%
	\clearfield{place}%
	\clearlist{location}%
	\clearfield{eventtitle}%
	\ifentrytype{inproceedings}{%
		\clearlist{publisher}
	}{}%
	\ifentrytype{software}{%
	}{%
		\clearfield{url}%
		\clearfield{urlyear}%
	}%
}

\newtheorem{example}{Example}

\newtheorem{lemma}{Lemma}
\newtheorem{theorem}{Theorem}

\newtheorem{fact}{Fact}

\addto\captionsenglish{%
}
\addto\extrasenglish{%
}

\addto\captionsenglish{%
}
\addto\extrasenglish{%
}

\addto\captionsenglish{%
  \def\figurename{Fig.\!}%
}
\addto\extrasenglish{%
}

   \makeatletter
    \newcommand{\linebreakand}{%
      \end{@IEEEauthorhalign}
      \hfill\mbox{}\par
      \mbox{}\hfill\begin{@IEEEauthorhalign}
    }
    \makeatother
\makeatletter
\newcommand{\newlineauthors}{%
  \end{@IEEEauthorhalign}\hfill\mbox{}\par
  \mbox{}\hfill\begin{@IEEEauthorhalign}
  }
  
\makeatother

\makeatletter
\def\tagform@#1{\maketag@@@{\ignorespaces#1\unskip\@@italiccorr}}
\let\orgtheequation\theequation
\def\theequation{(\orgtheequation)}
\makeatother
\let\orgautoref\autoref
\renewcommand{\autoref}[1]{\def\equationautorefname{Eq.}\orgautoref{#1}}
\renewcommand{\baselinestretch}{1.0}
\begin{document}

\title{%
Depth-Optimal Synthesis \\
of Clifford Circuits with SAT Solvers
}

\author{
  \IEEEauthorblockN{Tom Peham\IEEEauthorrefmark{1}}
\IEEEauthorblockA{\textit{Chair for Design Automation} \\
\textit{Technical University of Munich}\\
Germany \\
tom.peham@tum.de}
\and
\IEEEauthorblockN{Nina Brandl\IEEEauthorrefmark{1}}
\IEEEauthorblockA{\textit{Institute for Integrated Circuits} \\
\textit{Johannes Kepler University Linz}\\
Austria \\
nina.brandl@jku.at}
\newlineauthors
\IEEEauthorblockN{Richard Kueng}
\IEEEauthorblockA{\textit{Institute for Integrated Circuits} \\
\textit{Johannes Kepler University Linz}\\
Austria \\
richard.kueng@jku.at}
\and
\IEEEauthorblockN{Robert Wille}
\IEEEauthorblockA{\textit{Chair for Design Automation} \\
\textit{Technical University of Munich}\\
Germany \\
\textit{Software Competence Center Hagenberg GmbH}\\
Austria\\
robert.wille@tum.de}
\and
\IEEEauthorblockN{Lukas Burgholzer}
\IEEEauthorblockA{\textit{Institute for Integrated Circuits} \\
\textit{Johannes Kepler University Linz}\\
Austria \\
lukas.burgholzer@jku.at}

\linebreakand

\IEEEauthorblockA{\IEEEauthorrefmark{1} Both authors contributed equally to this work.}%
\vspace*{-5mm}
}

\maketitle

\begin{abstract}
Circuit synthesis is the task of decomposing a given logical functionality into a sequence of elementary gates. 
It is (depth-)optimal if it is impossible to achieve the desired functionality with even shorter circuits.
Optimal synthesis is a central problem in both quantum and classical hardware design, but also plagued by complexity-theoretic obstacles.
Motivated by fault-tolerant quantum computation, we consider the special case of synthesizing blocks of Clifford unitaries.
Leveraging entangling input stimuli and the stabilizer formalism allows us to reduce the Clifford synthesis problem to a family of poly-size satisfiability (SAT) problems -- one for each target circuit depth.
On a conceptual level, our result showcases that the Clifford synthesis problem is contained in the first level of the polynomial hierarchy ($\mathsf{NP}$), while the classical synthesis problem for logical circuits is known to be complete for the second level of the polynomial hierarchy ($\Sigma_2^{\mathsf{P}}$).
Based on this theoretical reduction, we formulate a SAT encoding for depth-optimal Clifford synthesis.
We then employ SAT solvers to determine a satisfying assignment or to prove that no such assignment exists. From that, the shortest depth for which synthesis is still possible (optimality) as well as the actual circuit (synthesis) can be obtained.
Empirical evaluations show that the optimal synthesis approach yields a substantial depth improvement for random Clifford circuits and Clifford+T circuits for Grover search.

\end{abstract}

\section{Introduction}
\label{sec:intro}

Quantum computing %
is a computational paradigm that might offer computational advantages over classical algorithms for certain problems.
State of the art quantum computing hardware is still limited in scale, featuring a relatively small number of qubits that are prone to errors.
While finding short-term applications for these noisy intermediate scale quantum (NISQ) computers is an ongoing and vibrant research field~\cite{preskillQuantumComputingNISQ2018}, they cannot be used for advanced applications such as integer factoring (Shor;~\cite{shorPolynomialtimeAlgorithmsPrime1997}), unstructured search (Grover;~\cite{groverFastQuantumMechanical1996}), solving linear systems (HHL;~\cite{harrowQuantumAlgorithmLinear2009}), or convex optimization~\cite{brandaoQuantumSpeedUpsSolving2017,apeldoornQuantumSDPSolversBetter2020,brandaoFasterQuantumClassical2022}.

To scale up quantum computing to longer computations with many qubits, \emph{fault-tolerant} computation schemes have to be used that protect the information of a qubit against errors and allow the application of quantum gates with a low error-rate, see e.g.~\cite{kitaevQuantumComputationsAlgorithms1997,shorFaulttolerantQuantumComputation1996,nielsenQuantumComputationQuantum2010}.
One way to do this is using quantum error correcting codes and performing all computations in the Clifford+T gate-set~\cite{gottesmanStabilizerCodesQuantum1997}.
The benefit of restricting the gate-set is that
every gate 
can be performed in a fault-tolerant fashion on an appropriate code.

Fault-tolerant quantum circuits can become quite large due to the overhead from breaking down every quantum computation into this restricted gate-set, with error syndrome extraction and error correction steps.
For near-term applications they are usually optimized to have a minimal two-qubit gate count as these gates tend to have the highest error rates. For fault-tolerant computations, a suitable performance metric is circuit depth because it directly correlates with the runtime of the computation.

\begin{figure*}[t]
  \centering
  \begin{subfigure}[b]{1.0\linewidth}
  \resizebox{\linewidth}{!}{
  \begin{tikzpicture}
  \begin{yquant*}
    qubit q[3];

    [this subcircuit box style={dashed, label=Clifford Block,fill=green!20}]
    subcircuit {
    qubit {} q[3];
    h q[0];
    x q[0];
    h q[1];
    h q[2];
    x q[2];
    cnot q[1]|q[2];
  } (q[0-2]);
    box {$T^\dagger$} q[1];
    cnot q[1]|q[0];
    box {$T$} q[1];
    cnot q[1]|q[2];
    box {$T^\dagger$} q[1];
    cnot q[1]|q[0];
    box {$T$} q[2];
    cnot q[2]|q[0];
    box {$T$} q[1];
    box {$T$} q[0];
    box {$T^\dagger$} q[2];
    
    [this subcircuit box style={dashed, label=Clifford Block,fill=cyan!20}]
    subcircuit {
      qubit {} q[3];
      x q[1];
      h q[1];
      cnot q[2]|q[0];
      cnot q[1]|q[0];
      x q[0];
      h q[0];
      x q[0];
      h q[1];
      x q[1];
      z q[1];
      h q[1];
      x q[1];
      x q[2];
      z q[2];
      h q[2];
      x q[2];
      cnot q[2]|q[1];
    } (q[0-2]);
  
    box {$T^\dagger$} q[2];
    cnot q[2]|q[0];
    box {$T$} q[2];
    cnot q[2]|q[1];
    box {$T$} q[1];
    box {$T^\dagger$} q[2];
    cnot q[2]|q[0];
    cnot q[1]|q[0];
    align q;
    box {$T$} q[0];
    box {$T^\dagger$} q[1];
    box {$T$} q[2];
    [this subcircuit box style={dashed, label=Clifford Block,fill=magenta!20}]
    subcircuit {
    qubit {} q[3];
    cnot q[1]|q[0];
    x q[0];
    h q[0];
    x q[1];
    h q[1];
    x q[2];
    h q[2];
  } (q[0-2]);
  \end{yquant*}
\end{tikzpicture}}
\caption{Unoptimized Circuit}\label{fig:grover_unoptimized}
\end{subfigure}

\begin{subfigure}[b]{1.0\linewidth}
    \resizebox{.863\linewidth}{!}{
    \begin{tikzpicture}
  \begin{yquant*}
    qubit q[3];

    [this subcircuit box style={dashed,fill=green!20}]
    subcircuit {
    qubit {} q[3];
    h q[0];
    x q[0];
    h q[1];
    h q[2];
    x q[2];
    cnot q[1]|q[2];
  } (q[0-2]);
    box {$T^\dagger$} q[1];
    cnot q[1]|q[0];
    box {$T$} q[1];
    cnot q[1]|q[2];
    box {$T^\dagger$} q[1];
    cnot q[1]|q[0];
    box {$T$} q[2];
    cnot q[2]|q[0];
    box {$T$} q[1];
    box {$T$} q[0];
    box {$T^\dagger$} q[2];
    
    [this subcircuit box style={dashed,fill=cyan!20}]
    subcircuit {
      qubit {} q[3];
      cnot q[2]|q[0];
      h q[1];
      x q[1];
      x q[0];
      cnot q[1]|q[0];
      h q[0];
      h q[2];
      z q[2];
      cnot q[2]|q[1];
    } (q[0-2]);
  
    box {$T^\dagger$} q[2];
    cnot q[2]|q[0];
    box {$T$} q[2];
    cnot q[2]|q[1];
    box {$T$} q[1];
    box {$T^\dagger$} q[2];
    cnot q[2]|q[0];
    cnot q[1]|q[0];
    align q;
    box {$T$} q[0];
    box {$T^\dagger$} q[1];
    box {$T$} q[2];
    [this subcircuit box style={dashed,fill=magenta!20}]
    subcircuit {
    qubit {} q[3];
    cnot q[1]|q[0];
    x q[0];
    h q[0];
    x q[1];
    h q[1];
    x q[2];
    h q[2];
  } (q[0-2]);
  \end{yquant*}
\end{tikzpicture}}
\caption{Optimized Circuit}\label{fig:grover_optimized}
\end{subfigure}
  \caption{\emph{Two Clifford+T circuits for 3-qubit Grover search:} the oracle corresponds to a random 3-SAT formula with 3 variables and 5 clauses. \newline
  \textbf{(top)} decomposition of one Grover block into Clifford+T gates, the actual synthesis is achieved using Qiskit's \texttt{Grover} class using a \texttt{PhaseOracle}~\cite{Qiskit-Textbook}. \newline
  \textbf{(bottom)} depth-optimized synthesis of all (nontrivial) Clifford blocks using the methods from this work. Our SAT approach certifies that the green and purple blocks are already depth-optimal. This is not the case for the central teal block, where our method yields considerable improvements (depth $9$ vs.\ depth $5$).
  }
  \label{fig:grover}\vspace*{-3mm}
\end{figure*}
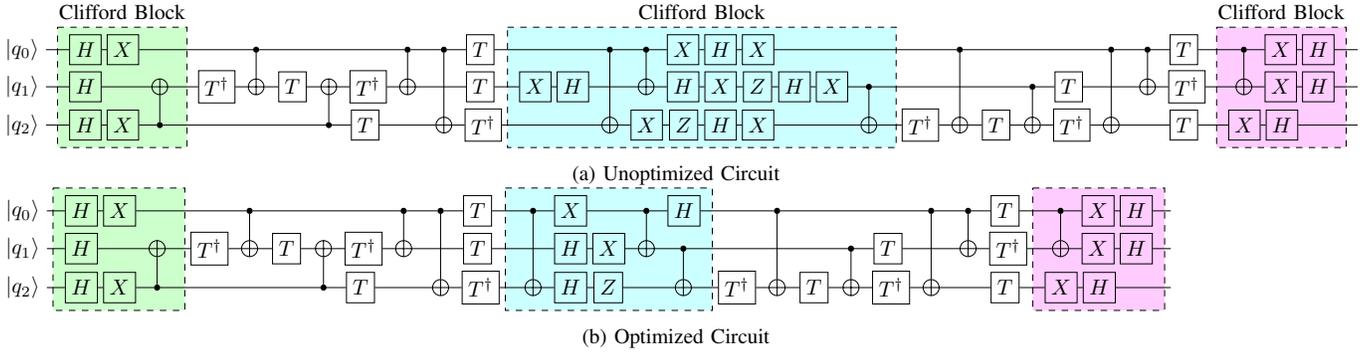

The problem of synthesizing optimal circuits in the fault-tolerant regime has commonly been considered holistically, trying to minimize the entire circuit, or partially, by reducing the T-gate count only, see e.g.~\cite{kliuchnikovSynthesisUnitariesClifford2013,amyMeetinthemiddleAlgorithmFast2013,dimatteoParallelizingQuantumCircuit2016,niemannAdvancedExactSynthesis2020}. 
The general synthesis problem turns out to be very hard.
Instead of minimizing the entire Clifford+T circuit, one can consider only the Clifford parts to try to make the general synthesis problem easier.

In this work, we show that the problem of synthesizing depth-optimal Clifford circuits is, in fact, at most as hard as Boolean satisfiability (SAT).
We achieve this by considering the synthesis problem for a larger circuit that receives a pairwise maximally entangled state as input.
Although this theoretically increases the problem size, the utilization of the maximally entangled state allows for breaking down the synthesis problem in terms of %
stabilizers (Gottesman-Knill).

Based on this perspective, we formulate a SAT encoding for synthesizing and optimizing $n$-qubit Clifford circuits of maximal depth $d_{\max}$ with $O(n^2d_{\max})$ variables and $O(n^4d_{\max})$ constraints.
For larger circuits, we furthermore develop heuristic optimization routines %
based on optimal SAT encodings that reduce Clifford circuit depth in a divide-and-conquer approach.

We implemented the proposed methods and compared them to a state of the art and openly available Clifford synthesis technique. 
Results show that while the optimal synthesis scales only up to 5 qubits, the state of the art is, on average, two orders of magnitude away from the optimal depth.
The heuristic approaches also yield better results than the state of the art and are much more scaleable than the optimal approach.
The efficacy of the methods for fault-tolerant quantum computations is illustrated by showing how the depth of Clifford+T circuits can be reduced by optimizing Clifford sub-circuits.
All implementations are publicly available via the quantum circuit compilation tool QMAP~\cite{willeMQTQMAPEfficient2023}, which is part of the \emph{Munich Quantum Toolkit} (MQT) and accessible at \mbox{\url{https://github.com/cda-tum/qmap}}.

The remainder of this work is structured as follows: \autoref{sec:background} motivates the need for synthesizing depth-optimal Clifford circuits and provides context to the classical circuit synthesis problem.
Then, \autoref{sec:related_work} briefly goes over previous work on quantum circuit synthesis, especially in the Clifford case.
In \autoref{sec:theory} we show the main construction of this work and the reduction of Clifford synthesis to SAT before giving the full details on the SAT encoding in \autoref{sec:sat}.
Based on that, we derive a heuristic optimization routine in \autoref{sec:heuristic}.
The evaluations of the methods introduced in this work are presented in \autoref{sec:evaluations}.
Finally, \autoref{sec:conclusion} concludes this work.

\section{Background and Motivation}
\label{sec:background}

To keep this work self-contained, the following sections review the central concepts of fault-tolerant quantum computing and circuit synthesis that are required throughout the rest of this work.

\subsection{Fault Tolerant Quantum Computing}
\label{subsec:qc}

Quantum computations are prone to errors from multiple sources, be it decoherence, information leakage or other kinds of noise.
In order to protect quantum information, \emph{quantum error correcting codes} (QECCs) have been introduced.
Essentially the information of a single \emph{logical qubit} is encoded into multiple \emph{physical qubits} for which various different encoding schemes have been discovered.

Executing a quantum gate on a logical qubit is not a straightforward matter anymore as it is not at all obvious how to extend the functionality of a physical gate to the logical level.
Furthermore, applying a gate to a logical qubit should be possible in a \emph{fault-tolerant} fashion, such that applying a gate should introduce errors that can be corrected with an arbitrarily high probability.

The \emph{Clifford gate-set} is a set of gates that can be applied \emph{transversally} to many QECCs, i.e., applying a Clifford operation to a logical qubit can be done by performing that Clifford gate on every physical qubit individually.
Therefore, if a gate application introduces an error onto a physical qubit, this error cannot spread throughout the encoded qubit and is simpler to correct.
Any Clifford unitary can be obtained from a sequence of \emph{Hadamard}, \emph{Phase} and \emph{CNOT} gates.
Notably, the Pauli $X$-, $Y$- and $Z$-gates are also in the Clifford group.

Clifford gates alone are not sufficient to perform universal quantum computation and can be classically simulated in polynomial time~\cite{gottesmanStabilizerCodesQuantum1997,aaronsonImprovedSimulationStabilizer2004}.
However, adding a \mbox{T-gate} \mbox{$T = |0 \rangle \! \langle 0|  + e^{i\frac{\pi}{8}} |1 \rangle \! \langle 1|$} suffices to achieve universality.
A \mbox{T-gate} can be realized fault-tolerantly using the magic state distillation protocol~\cite{bravyiUniversalQuantumComputation2005} and gate teleportation.

A fault-tolerant quantum computation can therefore be realized by alternating blocks of Clifford gates and blocks of T-gates where error syndromes are detected and corrected throughout the computation.
\autoref{fig:grover} shows the quantum circuit for a 3-qubit Grover search that is implemented using the Clifford+T gate-set on the logical qubits $\ket{q_0}$, $\ket{q_1}$ and $\ket{q_2}$.

In this fault-tolerant regime, the most important performance metric is the \emph{circuit depth}.
Assuming gates on different qubits can be executed in parallel, the depth directly determines the runtime of the computation.
As the \emph{cycle time} can vary considerably  depending on the technology the quantum computer is built upon, reducing the depth of a circuit can reduce the cost of a computation drastically.

\begin{example}
  The circuit in \autoref{fig:grover_unoptimized} has 3 Clifford blocks of interest.
  \autoref{fig:grover_optimized} shows the same circuit where each block was synthesized to have minimal depth. %
  As one can see, the two outermost blocks were already depth-optimal in the original circuit.
  Although this looks intuitively true, it still needs some form of proof. %
  The Clifford block in the middle, however, has a depth of $9$ (the first Hadamard and CNOT-gate are parallel) and can be optimized to depth $5$.
  High-depth Clifford blocks like this appear frequently when synthesizing circuits using state of the art synthesis tools for quantum computing\footnote{The circuit in question was generated using Qiskit 0.42.1}.
\end{example}
\noindent
The quest of finding such depth-optimal realizations of Clifford circuits will be our primary motivation in this work as the restricted nature of these circuits lends itself nicely to classical design automation methods such as \emph{SAT} and \emph{SMT solvers}.

\subsection{From Classical to Quantum Synthesis}
\label{sec:motivation}
Before diving into the details of our quantum protocol, it is worthwhile to briefly review classical equivalence checking of logical circuits (or functions), as well as classical circuit synthesis. This will provide guidance and also motivate our use of SAT solvers for quantum circuit synthesis.

\subsubsection{Classical Equivalence Checking and SAT}

Let $C,C'$ be two circuits with $n$ input bits. %
We say that these circuits are equivalent ($C \simeq C'$) if and only if they produce the same output for each conceivable input. In formulas, \mbox{$\forall x \in \left\{0,1\right\}^n\; Cx = C'x$}.
Taking the contrapositive yields
\begin{equation*}
C \not\simeq C' \quad \Leftrightarrow \quad 
\exists x \in \left\{0,1\right\}^n \; \text{s.t.}\; Cx \neq C'x.
\end{equation*}
The logical not-equality $Cx \neq C'x$ can readily be converted into a 
Boolean function $\phi_{C,C'}(x)$ with input $x$.
This highlights an interesting one-to-one correspondence between (the negation of) circuit equivalence and the satisfiability problem (SAT).
On a theoretical level, SAT is a hard problem that is complete for the problem class NP. %
Nevertheless a plethora of heuristic SAT solvers~\cite{alounehComprehensiveStudyAnalysis2019} %
perform very well in practice.

\subsubsection{Classical Optimal Synthesis and QBF}

Circuit synthesis %
aims to find a logical circuit, that implements a desired $n$-bit target functionality with as few elementary gates as possible. Here, we will focus on \emph{circuit depth} (i.e. layers of gates). A decision version of this optimization problem looks as follows:
\begin{equation}
\exists C_d, \mathrm{depth}(C_d) \leq d_{\max} \; \forall x \in \left\{0,1\right\}^n \; C_d x = C x,
\label{eq:circuit-synthesis-original}
\end{equation}
where $C_d$ is a placeholder for another logical circuit.
In words, this formula evaluates to true if and only if it is possible to exactly reproduce the functionality of circuit $C$ with another logical circuit $C_d$ that obeys $\mathrm{depth}(C_d) \leq d_{\max}$.
Multiple queries to this logical function with different values of $d_{\max}$ allow us to determine the optimal depth of \emph{any} circuit synthesis, e.g.\ via binary search.

It is also possible to rephrase \autoref{eq:circuit-synthesis-original} as a quantified Boolean formula (QBF). For starters, note that we can represent any logical circuit $C_d$ by a binary encoding %
 $y$ of length (at most) $\mathrm{poly}(nd_{\max})$. 
Different bit strings $y$ give rise to different circuits $C_d$ and vice versa. With this binary encoding at hand, we can adapt the Boolean function reformulation of equivalence checking to the case at hand: $\phi_C (y,x) =1$ if $Cx=C_dx$ and $C_d$ is the circuit encoded by $y$. Otherwise, this formula evaluates to $0$. Putting everything together, we obtain the following QBF reformulation of logical circuit synthesis:
\begin{equation}
\exists y \in \left\{0,1\right\}^{\mathrm{poly}(nd_{\max})} \; \forall x \in \left\{0,1\right\}^n \; \phi_C (y,x)\overset{!}{=}1.
\label{eq:circuit-synthesis}
\end{equation}
The exist quantifier ($\exists$) ranges over all possible logical circuit encodings with depth at most $d_{\max}$ ($C_d \leftrightarrow y$) while the forall quantifier ($\forall$)  ranges over all $2^n$ possible input bitstrings.

Such QBFs are, in general, much harder to handle than a mere SAT problem with only one type of quantifier. In fact, the reformulated circuit synthesis problem~\autoref{eq:circuit-synthesis} is complete for the problem class $\Sigma_2^p$ -- one branch of the second level of the polynomial hierarchy \cite{arora_barak_2009}.
Unless the polynomial hierarchy collapses to the first level (which is widely believed to be false), these problems are much harder than SAT and, by extension, equivalence checking.

QBFs do arise naturally in a variety of contexts~\cite{shuklaSurveyApplicationsQuantified2019}.
Solvers do exist, see e.g.~\cite{rabeIncrementalDeterminization2016,lonsingDepQBFSearchbasedQBF2017,rabeCAQECertifyingQBF2015} and typically rely on the counter-example-guided inductive synthesis principle (CEGIS)~\cite{solar-lezamaCombinatorialSketchingFinite2006} which has its roots in abstraction refinement (CEGAR)~\cite{clarkeCounterexampleguidedAbstractionRefinement2000,jhaTheoryFormalSynthesis2017}.
For program synthesis, for example, CEGIS-style solvers alternate between generating candidate programs and checking them for counter-examples.

\subsubsection{Going quantum: Circuit Equivalence and Synthesis}

The two classical challenges we just discussed have natural counterparts in the quantum realm.
Quantum circuits that act on $n$ qubits are also comprised of elementary (quantum) gates, but their functionality is radically different. 
We say that two such circuits $U,V$ are \emph{equivalent} if and only if
\begin{equation*}
U \simeq V \quad \Leftrightarrow \quad \forall |\psi \rangle \in \mathbb{C}^{2^n} \; U |\psi \rangle = \mathrm{e}^{i \phi} V |\psi \rangle,
\end{equation*}
where $\phi \in [0,2 \pi)$ is a complex phase.
In contrast to (classical) logical circuits, there are infinitely many possible input states $|\psi\rangle$ to be checked.\footnote{This is already our point of departure from earlier work, most notably~\cite{schneiderSATEncodingOptimal2023}. 
There, a subset of the authors asked a related, but simpler, question that arises from only considering $|\psi \rangle =|0,\ldots,0\rangle$ (i.e.\ fix a single input state).} Consistency and mutual interrelations permit us to compress this number to (at most) $4^n$  disjoint input states~\cite{grossQuantumStateTomography2010,klieschGuaranteedRecoveryQuantum2019,rothRecoveringQuantumGates2018,burgholzerRandomStimuliGeneration2021,gutaFastStateTomography2020}.
This number is, however, still exponential in the total number of qubits.
It is known that the (negated) unitary equivalence problem is QMA-complete~\cite{janzingNonidentityCheckQMAcomplete2005} (QMA is the appropriate quantum generalization of the classical problem class NP).
So, at face value, the quantum circuit equivalence problem looks even harder than its classical counterpart. 

Suppose that we are given a target unitary $U$, e.g. in the form of a high-level quantum circuit, and we want to decompose it into as few gate layers as possible (i.e. circuit depth). This practical problem occurs whenever we want to execute a high-level unitary (e.g. a quantum algorithm) on an actual $n$-qubit quantum computer. The details of this synthesis problem depend on the type of elementary gate-set, but always produces a two-fold quantified problem 
\begin{equation*}
\exists U_d, \mathrm{depth}(U_d) \leq d_{\max}\; \forall |\psi \rangle\; U_d |\psi \rangle = \mathrm{e}^{i \phi} U |\psi \rangle,
\end{equation*}
that resembles \autoref{eq:circuit-synthesis-original} with non-binary quantum states and a complex phase $\phi \in [0,2\pi)$. This problem is at least as hard as logical gate synthesis, because it includes (reversible) encodings of the latter one as a special case. 
Reversibility of quantum circuits allows us to slightly streamline this display:
\begin{equation}
\exists U_d, \mathrm{depth}(U_d) \leq d_{\max}\; \forall |\psi \rangle\; UU_d^\dagger |\psi \rangle = \mathrm{e}^{i \phi}|\psi \rangle, \label{eq:quantum-synthesis}
\end{equation}
where $U_d^\dagger$ is the adjoint or reverse circuit of $U_d$. %

For some special cases like Clifford circuits, we can recast \autoref{eq:quantum-synthesis} as a logical Boolean formula whose intrinsic difficulty is on par with logical circuit equivalence or SAT as shown in \autoref{sec:theory}.
The difficulty level is then \emph{much} lower than general classical circuit synthesis. Empowered by this rigorous theoretical insight, we then employ state of the art SAT solvers to address the full Clifford synthesis problem.  

\section{Related work}\label{sec:related_work}

Quantum circuits comprised of universal gate-sets are universal in the sense that they can approximate \emph{every} unitary evolution. Quantum gate synthesis can be viewed as a quantitative take on this very issue: what is the best way to implement a given unitary evolution, e.g. a quantum computation?

The celebrated Solovay-Kitaev theorem~\cite{kitaevQuantumComputationsAlgorithms1997,dawsonSolovayKitaevAlgorithm2006} can be viewed as a very general synthesis protocol for arbitrary single-qubit unitaries ($n=1$) and arbitrary universal gate-sets. 
Extensions to $n$ qubits can be achieved by either combining single-qubit gate synthesis with certain entangling multi-qubit gates, or by directly generalizing the Solayev-Kitaev algorithm to $d=2^n$-dimensional unitaries~\cite{dawsonSolovayKitaevAlgorithm2006}. 
Although implementations do exist, see e.g.~\cite{phamOptimizationSolovayKitaevAlgorithm2013}, this synthesis algorithm is typically too slow for practical purposes. As a result, the community has moved away from this rigorous meta-algorithm and towards more scalable heuristics. 
Many of them address the universal gate-set comprised of Clifford+T gates, see e.g.\ ~\cite{kliuchnikovSynthesisUnitariesClifford2013,amyMeetinthemiddleAlgorithmFast2013,dimatteoParallelizingQuantumCircuit2016,niemannAdvancedExactSynthesis2020}.

Clifford circuits (without T gates) are an interesting special case in this context. Mathematically speaking, they form a representation of a finite symplectic group~\cite{dehaeneCliffordGroupStabilizer2003,grossHudsonTheoremFinitedimensional2006,zhuCliffordGroupFails2016} with additional structure. For instance, it is known that every Clifford unitary can be decomposed into a Clifford circuit of depth at most $O(n)$~\cite{maslovLinearDepthStabilizer2007}.
Such insights highlight that Clifford circuits cannot be overly complex -- a feature that
also extends to Clifford synthesis. For the special case of $n=6$ qubits, competitive Clifford synthesis protocols have been put forth in \cite{bravyi6qubitOptimalClifford2020}. 
For general $n$, the algorithms by Koenig and Smolin can be used to associate a given Clifford circuit with exactly one element of the Clifford group. And, subsequently, their algorithm can be used to synthesise this very group element. 
More recently, the group of Robert Calderbank developed Clifford synthesis algorithms that even work on the logical level (i.e.\ on top of an error correcting stabilizer code)~\cite{rengaswamyLogicalCliffordSynthesis2020}, while Bravyi \emph{et al.} discovered constant depth representations of arbitrary Clifford circuits, under the assumption that  one is allowed to use global entangling operations of Ising type~\cite{bravyiConstantcostImplementationsClifford2022}.  

While the aforementioned approaches do produce a provably correct decomposition of Clifford circuits into elementary (Clifford) gates, it is not clear whether size and depth are 
(close to) optimal. This is where reformulations in terms of SAT/QBF can make a significant difference. They reformulate the gate synthesis problem as a family of quantified Boolean formulas (QBFs), one for each maximum circuit depth $d_{\max}$ we allow.
Such a QBF evaluates to true if and only if an exact circuit representation with depth (at most) $d_{\max}$ is possible and returns the explicit representation. Otherwise, it evaluates to false.
Attempting to solve these QBFs for different depths bears the potential of identifying the best circuit representation of a given functionality. This is why SAT/QBF-based synthesis approaches have long been a mainstay in classical design automation~\cite{bloemSATBasedMethodsCircuit2014}.
In fact, the idea of combining SAT-based synthesis with Clifford circuits is not entirely new. In Ref.~\cite{schneiderSATEncodingOptimal2023}, a subset of authors proposes this very idea for optimal stabilizer state preparation: find the shortest Clifford circuit that takes $|0,\ldots,0\rangle$ as input and produces a known target stabilizer state. The ideas presented here may be viewed as an extension of these earlier ideas to full Clifford circuit synthesis. 
In addition, we also supply rigorous proofs of correctness and provide additional context, as well as background.

\section{Main result and theoretical underpinning} \label{sec:theory}

Our main conceptual result is a one-to-one correspondence between Clifford synthesis, on the one hand, and Boolean satisfiability (SAT) on the other. 
The main result is displayed in \autoref{thm:main} and originates from two genuinely quantum twists to the original synthesis questions: (i) maximally entangled input stimuli and (ii) the Gottesman-Knill theorem. It forms a rigorous foundation for optimal Clifford circuit synthesis with SAT, the topic of \autoref{sec:sat} below.

\subsection{Quantum twist 1: maximally entangled input states}

The first step in our theoretical argument goes by many names, including the Choi-Jamiolkowski isomorphism~\cite{choiCompletelyPositiveLinear1975, jamiolkowskiLinearTransformationsWhich1972}, entanglement-assisted process tomography~\cite{altepeterAncillaassistedQuantumProcess2003} 
 and the flattening operation in tensor analysis~\cite{watrousTheoryQuantumInformation2018,kuengACM270Quantum2019}.
Conceptually we take two quantum circuits $U$ and $V$ on $n$-qubits. Instead of testing all possible input states, we create a single universally valid input state to check their equivalence. In a first thought process, we consider a test circuit twice as large and apply $U$ to the top $n$ qubits and the inverse of each gate in $V$ to the remaining $n$ qubits. By entangling the $k$th input qubit of $U$ with the $k$th input qubit of $V$ we create a new state $|\omega_{2n} \rangle$ of size $2n$. As $U$ and $V$ operate on entangled qubits, all changes applied by the first circuit will be reverted by the second circuit only if and only if they have the same functionality up to a global phase. If they instead differ at any point of their unitary, the resulting state will not be equal to $|\omega_{2n} \rangle$ again.
The idea of two circuits with identical unitaries reverting the changes of each other still applies when we let
$U$ and $V^\dagger$ both work on the same $n$ qubits, i.e. the control qubits of the pairwise entanglement. On the remaining target qubits we apply the identity as shown in \autoref{fig:entangle_circuit}.

\begin{figure}
    \centering
    \includegraphics[width=\columnwidth]{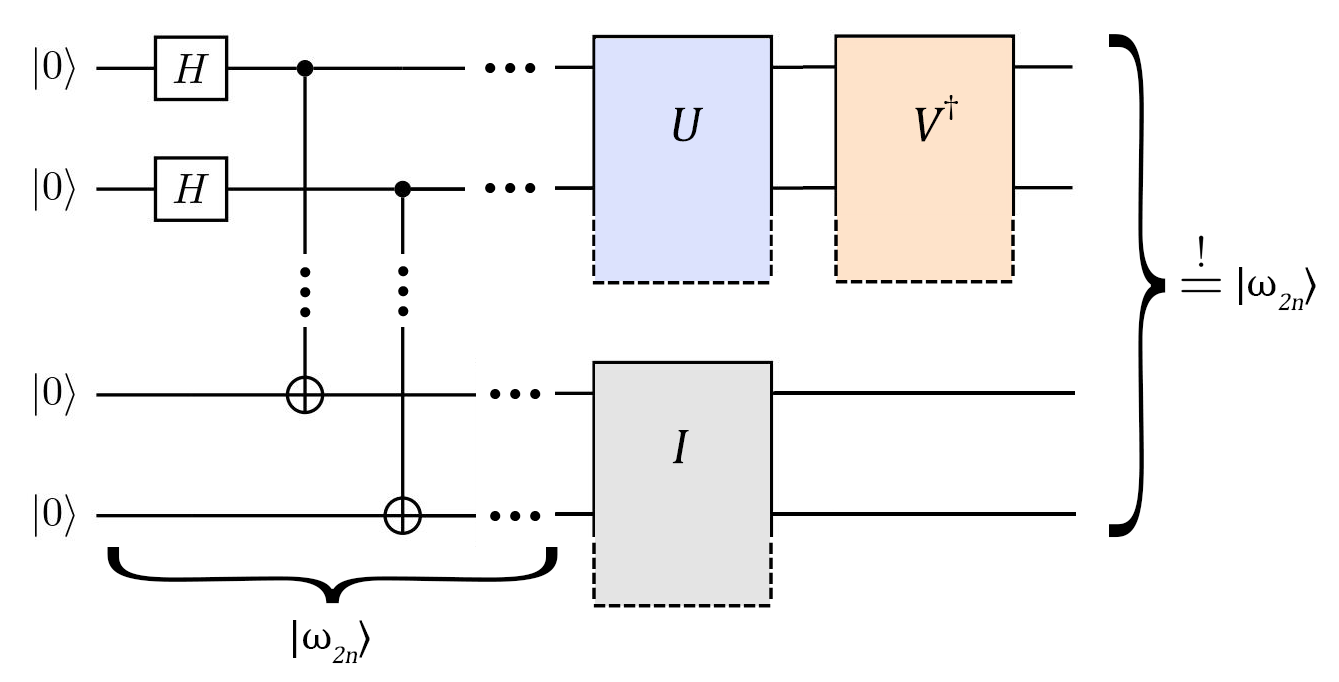}
    \caption{\emph{Illustration of entanglement-assisted equivalence-checking:}
    Two $n$-qubit circuits $U,V$ have equivalent functionality (up to a global phase), if and only if the above circuit produces the pairwise maximally entangled state $|\omega_{2n} \rangle$. Here, $V^\dagger$ is the reverse  circuit (adjoint) of $V$ and $I$ is the identity.
    }
    \label{fig:entangle_circuit}
\end{figure}

\begin{lemma} \label{lem:entanglement}
Let $U,V$ be two $n$-qubit quantum circuits and let $I^{\otimes n}$ be the $n$-qubit identity operation. Then, $U \simeq V$ if and only if
\begin{equation}
(UV^\dagger) \otimes I^{\otimes n} |\omega_{2n}\rangle \! \langle \omega_{2n}| (UV^\dagger)^\dagger \otimes I^{\otimes n} = |\omega_{2n} \rangle \! \langle \omega_{2n}|.
\label{eq:entanglement}
\end{equation}
here $|\omega_{2n} \rangle$ is the tensor product of $n$ 2-qubit Bell states that entangle the $k$th qubit with the $n+k$th qubit for all $k \in [n]$.
\end{lemma}

\begin{proof}%
Note that the Bell state is proportional to the vectorized identity matrix, i.e.
$|\omega_{2n}\rangle \sqrt{2^{n}}=\mathrm{vec}(I^{\otimes n})$. 
The operator-vector correspondence, see e.g. \cite[Eq.~(1.132)]{watrousTheoryQuantumInformation2018}, then asserts $UV^\dagger \otimes I_n |\omega_{2n} \rangle 
= \mathrm{vec}(UV^\dagger)/\sqrt{2^n}$.
This ensures that  $UV^\dagger=\mathrm{e}^{i \phi}I_n$ (as a matrix-valued equality) if and only if %
\begin{equation*}
UV^\dagger \otimes I |\omega_{2n} \rangle = \tfrac{1}{\sqrt{2^n}}\mathrm{vec}(UV^\dagger) \overset{!}{=} \tfrac{\mathrm{e}^{i \phi}}{\sqrt{2^n}}\mathrm{vec}(I^{\otimes n}) = \mathrm{e}^{i \phi}|\omega_{2n} \rangle.
\end{equation*}
\autoref{eq:entanglement} takes the outer product of this vector-valued equation to absorb the remaining global phase (mixed state formalism).
\end{proof}

Applying \autoref{lem:entanglement} to the 
quantum circuit synthesis formula~\autoref{eq:quantum-synthesis} 
gets rid of the complex phase and the forall quantifier over all possible input states:
\begin{equation}
\exists U_d,\mathrm{depth}(U_d) \leq d_{\max}\; (UU_d^\dagger) \otimes I \Omega (UU_d^\dagger)^\dagger \otimes I = \Omega,
\label{eq:quantum-synthesis-simplified}
\end{equation}
where we have used $\Omega$ as a shorthand notation for $|\omega_{2n} \rangle \! \langle \omega_{2n}|$.
However, for the erasure of an entire forall quantifier, we have to go to the mixed state formalism and effectively double the number of qubits involved from $n$ to $2n$, see \autoref{fig:entangle_circuit}.

\subsection{Quantum twist 2: Gottesman-Knill theorem}

On first sight, \autoref{eq:quantum-synthesis-simplified} looks much less daunting than \autoref{eq:quantum-synthesis} or even its classical counterpart~\autoref{eq:circuit-synthesis-original}. This is due to the fact that entanglement allows us to check for equivalence with a single input state (\autoref{lem:entanglement}) instead of a forall ($\forall$) over exponentially many possibilities.

To exploit this reformulation, there are two broad avenues on how to proceed:
(i) Use an actual quantum computer to empirically check the single remaining equivalence in \autoref{eq:quantum-synthesis-simplified}. This avenue was taken, for instance, in quantum assisted quantum compiling~\cite{khatriQuantumassistedQuantumCompiling2019}.
(ii) Use (strong) classical simulation of quantum circuits to check whether the two $2n$-qubit states in \autoref{eq:quantum-synthesis-simplified} are really equivalent. 
Four broad simulation approaches come to mind: array-based~\cite{deraedtMassivelyParallelQuantum2019}, stabilizer-based~\cite{aaronsonImprovedSimulationStabilizer2004}, tensor networks~\cite{bridgemanHandwavingInterpretiveDance2017} and decision diagrams~\cite{hillmichJustRealThing2020,niemannQMDDsEfficientQuantum2016}.

Here, we follow the second avenue and adopt stabilizer-based simulation. 
The main reason for this is 
that $\Omega =  |\omega_{2n} \rangle \! \langle \omega_{2n}|$, the $2n$-qubit state responsible for all these simplifications, is itself a stabilizer state with very desirable structure.

\begin{fact} \label{fact:bell-stabilizer}
The $2n$-qubit state $\Omega=|\omega_{2n} \rangle \! \langle \omega_{2n}|$ is a stabilizer state with generators $\langle(X \; I^{\otimes (n-1)})^{\otimes 2}, (Z \; I^{\otimes (n-1)})^{\otimes 2}\rangle$ where the parenthesis $(\;)$ stand for cyclic permutation.
\end{fact}

\noindent
It is well-known that the $2$-qubit Bell state is a stabilizer state with generators $\langle{XX, ZZ}\rangle$~\cite{nielsenQuantumComputationQuantum2010}. \autoref{fact:bell-stabilizer} follows from taking the $n$-fold tensor product of these generators at appropriate qubit locations. For $n{=}2$ for example this means $\langle XIXI, ZIZI, IXIX, IZIZ\rangle$ or for $n{=}3$ $\langle{XIIXII, ZIIZII, IXIIXI, IZIIZI, IIXIIX, IIZIIZ}\rangle$, etc. Still, every $k$th qubit is entangled with the $(n+k)$th qubit.

We can efficiently simulate stabilizer circuits (of which Clifford circuits are a part) in polynomial time on a classical computer according to the Gottesman-Knill theorem~\cite{gottesmanHeisenbergRepresentationQuantum1998}. To simulate the circuit's action on the $2n$-qubit stabilizer state $\Omega$, we perform a logical mapping of stabilizers depending on the gates and again obtain a stabilizer state. %
We can conclude whether the applied circuits performed the identity based on the equality of the input and output generators.

\begin{fact}[Gottesman-Knill]
Suppose that both $U$ and $U_d$ are $n$-qubit Clifford circuits. Then, it is possible to efficiently check
\begin{equation}
(UU_d^\dagger) \otimes I \Omega (UU_d^\dagger)^\dagger \otimes I = \Omega. \label{eq:gottesman-knill-fact}
\end{equation}
\end{fact}

\noindent
This also follows from the conceptual idea of the entangled input: If $U$ and $U_d$ have the same functionality, $U_d^\dagger$ would invert the stabilizer mapping done by $U$, which results in not altering the input stabilizers at all. Note that input and output generators must match exactly. As the bottom half of this $2n$ circuit applies the identity, it will not alter the generators on the last $n$ qubits. Therefore we can also cut the generators we need to check in half, eliminating the $^{\otimes 2}$ operation in \autoref{fact:bell-stabilizer}.

More precisely, it is possible to construct a Boolean function $\phi_{U,\Omega}(U_d)$ that evaluates to $1$ if the input circuit $U_d$ achieves \autoref{eq:gottesman-knill-fact} and $0$ otherwise. We will present such an explicit construction in \autoref{sec:sat} below. 
All that matters at this point is that 
we can represent any (at most) depth-$d_{\max}$ Clifford unitary $U_d$ with a bitstring $y \in \left\{0,1\right\}^{l}$ that contains (at most)
$l(n,d_{\max})=O \left( n^2 d_{\max}\right)$ Boolean variables. And, what is more, we can actually represent $\phi_{U,\Omega}(y)$ as a CNF with $O\left(n^4 d_{\max}\right)$ clauses of constant length.
Putting all this together ensures that we can rewrite the Clifford circuit synthesis problem as
\begin{equation}
\exists y \in \left\{0,1\right\}^{l(n,d_{\max})} \; \phi_{U,\Omega}(y) \overset{!}{=}1. \label{eq:clifford-synthesis-final}
\end{equation}

\subsection{Main result and synopsis}

The insights culminating in \autoref{eq:clifford-synthesis-final} are worth a prominent display and a bit of additional context.

\begin{theorem}[SAT reformulation of Clifford synthesis] \label{thm:main}
Let $U$ be a $n$-qubit Clifford unitary (target) and fix a maximum depth $d_{\max} \in \mathbb{N}$.
Then, the decision problem ``is it possible to exactly reproduce $U$ with (at most) $d_{\max}$ Clifford layers?''
can be rephrased as an instance of SAT with $O(n^2 d_{\max})$ variables and $O(n^4 d_{\max})$ clauses of constant size each.
\end{theorem}

\noindent
This insight has both conceptual and practical implications, especially if one keeps in mind that $d_{\max} \leq O(n)$ for any Clifford circuit~\cite{maslovLinearDepthStabilizer2007}. In turn, \emph{binary search} allows for exactly determining the minimum circuit depth for a given Clifford unitary $U$ by solving (at most) $\lceil \log_2 (n)\rceil+O(1)$ SAT reformulations for varying ansatz depths $d_{\max}$.
What is more, satisfying assignments of the Boolean formula in \autoref{eq:clifford-synthesis-final} are bit encodings
of actual Clifford circuits that exactly reproduce $U$ and have depth at most $d_{\max}$ (synthesis). %

On a conceptual level, this approach provides a poly-time reduction of optimal Clifford synthesis to (logarithmically many instances of) SAT. This highlights that this special case is \emph{much} easier than the general circuit synthesis problem (classical and quantum). In particular: optimal Clifford synthesis is at most as hard as classical circuit equivalence checking. %

On a practical level, \autoref{thm:main} provides a rigorous and context-specific motivation for employing state of the art SAT solvers to rigorously address the Clifford synthesis problem.
An actual step-by-step introduction to this encoding, as well as benchmarks, are the content of the remainder of this article.

\section{Optimal Clifford Circuit Synthesis with SAT}
\label{sec:sat}

The efficient simulability of Clifford circuits is based on the \emph{stabilizer tableau} encoding of a stabilizer state, see e.g.\ \cite{nielsenQuantumComputationQuantum2010} and references therein.
This polynomial-sized representation of a quantum state is the key to deriving a polynomially-sized SAT encoding for the considered synthesis problem.
Hence, before going into details on the encoding itself, we will give a brief recap on how to work with stabilizer tableaus.

\subsection{Stabilizer Tableau Representation of Stabilizer States}
\label{subsec:tableau}
An $n$-qubit stabilizer state can be represented by a \mbox{$(2n+1) \cross n$} binary matrix called the stabilizer tableau.
The idea is that every stabilizer generator for a state can be written using $2n+1$ bits of information.
In the standard notation~\cite{aaronsonImprovedSimulationStabilizer2004} for stabilizer tableaus, there are binary variables for Pauli $X$- and $Z$-type stabilizers $x_{i,j},z_{i,j}$ with $i,j \in \{0,1,..n-1\}$ and $r_i=1$ if the $i$-th stabilizer has a negative phase: %
$$
 \left[\begin{array}{ccc|ccc|c}
     x_{0,0} & \ldots & x_{0,n-1} & z_{0,0} & \ldots & z_{0,n-1} & r_0 \\
     \vdots & \ddots & \vdots & \vdots & \ddots & \vdots & \vdots \\
     x_{n-1,0} & \ldots & x_{n-1,n-1} & z_{n-1,0} & \ldots & z_{n-1,n-1} & r_{n-1}
\end{array}\right]$$
For a Pauli-$Y$ type stabilizer at position $i,j$ both $x_{i,j}$ and $z_{i,j}$ must be set to $1$ at the corresponding position.
The stabilizer states in this format can be altered by the usual Clifford gates and the following update rules, where $\oplus$ denotes a bitwise XOR operation:
\begin{itemize}
    \item \emph{Applying H on qubit $j$}: swaps the $j$th $X$-type column with the $j$-th $Z$-type column and $r\mathrel{\oplus}=x_jz_j$. This follows from the transformations $HXH^\dagger = Z$, $HZH^\dagger = X$ and $HYH^\dagger = -Y$, hence it switches $x_{i,j}\leftrightarrow z_{i,j} \; \forall i$ and flips the phase only in case of a Pauli-$Y$ stabilizer;
    \item \emph{Applying S on qubit $j$}: is a bitwise $XOR$ of the $j$-th \mbox{$X$-type} column to the $j$-th $Z$-type column $z_j\mathrel{\oplus}=x_j$ again with $r\mathrel{\oplus}=x_jz_j$;
    \item \emph{Applying CNOT on control qubit $c$ and target qubit $t$}: is a bitwise XOR of the $c$-th $X$-type column to the \mbox{$t$-th} $X$-type column $x_t\mathrel{\oplus}=x_c$ and vice versa for $Z$-type $z_c\mathrel{\oplus}=z_t$ as well as $r\mathrel{\oplus}=x_cz_t (x_t\oplus z_c\oplus1)$.
\end{itemize}
Further update rules for any other Clifford gate can be derived from these basic rules.

More information about a stabilizer state can be encoded into the tableau by including its \emph{destabilizers}~\cite{aaronsonImprovedSimulationStabilizer2004}.
These are Pauli-strings that, together with the stabilizers, generate the entire Pauli group.
They are treated identically to the stabilizer generators for the purpose of updating the stabilizer tableau.

\subsection{Tableau and Gate Variables}
\label{subsec:vars}

In the following, let $Q$ be the set of qubits acted on by a quantum circuit and $d_\mathit{max}$ the maximal depth of the circuit.

While all Clifford unitaries can be obtained from just $H$, $S$ and $CNOT$ gates, conveniently the target gate-set used for compilation may also include other gates like the Pauli $X$, $Y$, and $Z$ operations or two-qubit gates like the $CZ$ gate.
To reflect this flexibility in the encoding, we define two sets $\mathit{SQGs}$ and $\mathit{TQGs}$, the set of single-qubit gates and two-qubit gates, respectively, such that they can be used to implement any Clifford circuit.

At every layer of the quantum circuit, a certain gate can either be applied or not.
This suggests introducing the variables
\begin{alignat*}{3}
 Svars  & = \{g_q^d                           & \mid \;& g \in \mathit{SQGs},\; q \in Q,\; 0 \leq d < d_\mathit{max}\}                           \\
  Tvars & = \{g_{q_0,\; q_1}^d & \mid \; & g \in \mathit{TQGs},\; q_0 \in Q,\; q_1 \in Q\setminus\{q_0\}, \\
        &                                     &        & 0 \leq d < d_\mathit{max}\} 
\end{alignat*}
representing the application of a gate to a specific qubit (or pair of qubits) at depth $d$.

The possible stabilizer tableaus are encoded in a straightforward fashion according to their definition.
The $Z$-, $X$- and $R$-part of the tableau use the variables
\begin{alignat*}{2}
  Xvars & =\{z_q^d && \mid q \in Q,\; 0 \leq d < d_\mathit{max}\},  \\
  Zvars & =\{x_q^d && \mid q \in Q,\; 0 \leq d < d_\mathit{max}\},  \\
  Rvars & =\{r^d   && \mid 0 \leq d < d_\mathit{max}\},
\end{alignat*} 
where every element of the sets $Zvars$, $Xvars$ and $Rvars$ is a bitvector.
These bitvectors encode how all stabilizers act on the $Z$-, $X$-, $R$- part for a particular qubit.

Based on the construction in \autoref{lem:entanglement}, this encoding requires $2n$ qubits in order to guarantee that all circuits synthesized from these variables have the same unitary.
But having these $n$ additional qubits has the undesirable side-effect that the synthesized circuit should act as the identity on the lower $n$ qubits of the circuit.
This unnecessarily blows up the search space as the identity can be implemented ambiguously.
One could enforce constraints on these qubits, but this would unnecessarily increase the size of the encoding.
We can avoid this complication by considering only the upper $n$ qubits and switching from stabilizers of the entangled input state to stabilizers and destabilizers of the $\ket{0}^{\otimes n}$ state as stated by the following fact.

\begin{fact}
    For a Clifford unitary $U$ the stabilizers of $(U\otimes I^{\otimes n}) |\omega_{2n} \rangle$ on the first $n$ qubits are identical to the stabilizers and destabilizers of $U\ket{0}^{\otimes n}$.
\end{fact}

\noindent
Together with \autoref{fact:bell-stabilizer}, \autoref{lem:entanglement} tells us that for a Clifford circuit $U$ the $2n$ stabilizers of $U \otimes I_n |\omega_{2n} \rangle$ uniquely fix the unitary of the circuit.
Given $U$, we can explicitly calculate these stabilizers by propagating the generators for $\Omega$ through $U \otimes I$ ignoring the lower $n$ qubits.
This boils down to only analysing the first half of every stabilizer generator of the $2n$-qubit state $\Omega$. The result is a $2n(n+1)$ tableau for every given $U$. The initial tableau has diagonal entries with value $1$ and coincides with the stabilizers of the $|0\rangle^{\otimes n}$ input state ($Z$-type) combined with the respective destabilizers ($X$-type).

Hence, we can encode our problems using only $|Q|=n$ qubits and each bitvector for the tableau variables has size $2n$ since the information about the destabilizers has to be included as well.
\subsection{Transition Relation}
\label{subsec:transition}

With the variables defined, we can now encode how gates act on the stabilizer tableaus as described in \autoref{subsec:tableau}.
Naturally, the transition between tableaus would then be a constraint along the lines of
$$g_q^d \Rightarrow (\mathit{UpdateZ}(g, q, d) \land \mathit{UpdateX}(g, q, d) \land \mathit{UpdateR}(g, q, d)),$$
where the update formulas on the right encode the action of the gate on a qubit at a certain depth.
For a Hadamard this would mean\vspace*{-.5mm}
\begin{align*}
  \mathit{UpdateZ}(H, q, d) &= (z_q^{d+1} \Leftrightarrow x_q^{d}) \\
  \mathit{UpdateX}(H, q, d) &= (x_q^{d+1} \Leftrightarrow z_q^{d}) \\
  \mathit{UpdateR}(H, q, d) &= (r^{d+1} \Leftrightarrow (r^d \oplus x_q^{d}\land z_q^d))
\end{align*}\vspace*{-1.5mm}

While this encoding is correct, it is also wasteful in the sense that it would
lead to $|\mathit{SQGs}|+|\mathit{TQGs}|$ number of implications of this type for every qubit and depth.%
This number can be decreased significantly by noting that many gates act identically on the different parts of the stabilizer tableau (quantum computation is local).
The Pauli gates, for example, act as identity on the $Z$- and $X$-part of the tableau, only differing in how they change the $R$-part.

Since we know our gate-set, we can collect all possible transformations of the individual parts of the tableau a-priori.
Let $Z\textrm{--}updates(q, d)$ be the set of all possible updates to the $Z$-part of the stabilizer tableau on qubit $q$ at depth $d$.
The elements of these sets are logical formulas over the tableau variables.
We can then define a mapping $Z\textrm{--}impliedby(q, d): Z\textrm{--}updates(q,d) \rightarrow \mathcal{P}(Svars)$ that maps every update formula to the set of single-qubit gate variables that act on the stabilizer tableau with that update rule.

The single-qubit changes to the $Z$-part are then encoded by introducing the following constraint for every qubit $q$, depth $0 \leq d < d_\mathit{max}-1$ and $Z\textrm{--}update \in Z\textrm{--}updates(q, d)$:
$$\bigvee_{g_q^d \in Z\textrm{--}impliedby(Z\textrm{--}update)} \implies (z_q^{d+1} \Leftrightarrow Z\textrm{--}update).$$
Obviously, this can be done in a similar fashion for the $X$- and $R$-parts of the tableau as well as for two-qubit gates.

While the constraints at this point encode all possible stabilizer tableaus, there are some variable assignments that lead to invalid circuits, e.g. when a qubit is acted on by two gates at the same depth.
We, therefore, need to introduce another set of constraints for every depth $d$ and qubit $q$ to ensure consistency of the obtained solution.
\[
  \mathit{ExactlyOne}
  \begin{pmatrix*}[l]
    \{g_q^d &\mid g \in \mathit{SQGs}\} &\cup\\ \{g_{q,\;q_1}^d&\mid q_1 \in Q, g \in \mathit{TQGs}\} &\cup\\ \{g_{q_0,\;q}^d&\mid q_0 \in Q, g \in \mathit{TQGs}\}    
  \end{pmatrix*}
\]\vspace*{-2mm}

\subsection{Symmetry Breaking}
\label{subsec:symmetry}\vspace*{-.5mm}

Symmetry breaking~\cite{biereHandbookSatisfiability2009} is a widespread technique from the SAT-solving community. It introduces additional constraints to an existing CNF formula to avoid searching in symmetric parts of the search space.
This can be done in an automated fashion by analyzing the formula for automorphisms to obtain.
so-called \enquote{symmetry breakers}.
Doing this automatically has the downside that it is not clear which symmetries are found and whether the deduced constraints actually make the SAT instance any easier to solve or even harder.
In the case of the SAT formulation above, we can obtain symmetry breakers manually by using knowledge specific to Clifford synthesis.

We can impose additional constraints on the SAT solver by eliminating valid solutions that could be expressed in a simpler manner.
For example, the Hadamard gate is self-inverse.
We can therefore add the constraint
\vspace*{-.5mm}
\[H_q^d \implies \neg H_q^{d+1}\]
\vspace*{-1mm}
for $q \in Q$ and $0 \leq d < d_\mathit{max}-1$.
This eliminates all assignments to the gate variables that model a sequence of two consecutive Hadamard gates.

Another symmetry addresses possible degrees of freedom in the gate ordering and is best illustrated by means of the following equivalent circuits:
\vspace*{-1mm}
\begin{center}
\begin{tikzpicture}
  \begin{yquant}
    qubit q[2];

    h q[0];
    x q[1];
    box {$S$} q[1];
    cnot q[1] | q[0];
  \end{yquant}
\end{tikzpicture}
\qquad
\begin{tikzpicture}
  \begin{yquant}
    qubit q[2];

    x q[1];
    align q;
    h q[0];
    box {$S$} q[1];
    cnot q[1] | q[0];
  \end{yquant}
\end{tikzpicture}
\; .\vspace*{-1mm}
\end{center}
The Hadamard on the first qubit can either be parallel to the $X$ gate or the $S$ gate on the second qubit.
We can break this symmetry by imposing that the identity \mbox{single-qubit} gate cannot be followed by a non-identity single-qubit gate.
Otherwise, the non-identity gate could be moved to the left without changing the Clifford unitary.
More formally, for $q \in Q$, $0 \leq d < d_\mathit{max}-1$ we impose
\vspace*{1mm}
\[I_q^d \implies \bigwedge_{g \in \mathit{SQGs}\setminus\{I\}}\neg g_q^{d+1},\]
where $I$ is the identity gate.

A similar constraint can be imposed on two-qubit gates.
If the identity is applied to a pair of qubits, no two-qubit gate can come after the identities.
Again, we add the constraint
\[(I_{q_0}^d \land I_{q_1}^d) \implies \bigwedge_{g \in \mathit{TQGs}}\neg g_{q_0,\; q_1}^{d+1}.\]
for $q_0 \in Q$, $q_1 \in Q\setminus\{q_0\}$, $0 \leq d < d_\mathit{max}-1$.

There are many more symmetries that can be broken in the encoding.
In principal any Clifford gate identity can be used to derive a symmetry breaking constraint.
However, a trade-off has to be made between the number of constraints and the size of the solution space.

\vspace*{1mm}
\subsection{Optimizing Circuit Depth}
\label{subsec:cost}

The above encoding can be used to synthesize circuits that have \emph{at most} a depth of $d_{\max}$ but does not necessarily synthesize depth-optimal circuits.
This is only guaranteed if $d_{\max}$ is exactly the optimal depth, which has to be determined first.

One approach would be to start with an initial guess for $d_\mathit{max}$ and iteratively decrease it until the corresponding SAT instance has a solution but the SAT instance for $d_\mathit{max}-1$ does not.
A (theoretically) efficient way to achieve just that is binary search. 
The original circuit's depth can be used as an upper bound.
If this is far away from the optimum, however, it can lead to the generation of instances that are tough to solve.

Instead, the upper bound can be determined dynamically.
For a related problem (state preparation circuits), binary search has been explored in Ref.~\cite{schneiderSATEncodingOptimal2023} where it was also proposed to geometrically increase the depth horizon in which a solution is searched for.
Unfortunately, after a few iterations, the SAT calls will be quite costly and only promises to speedup the entire optimization if the runtime to solve a SAT instance grows sub-exponentially.
In case of exponential growth, simply searching linearly or in an arithmetic progression is faster.

Yet another way of gauging the initial depth is from empirical knowledge.
If it is known that, on average, the SAT method produces solutions that are 20\% shallower than the output circuit by another optimization routine, we can simply try to run it and start with the expected depth as an initial guess.
From this initial guess linear or quadratic probing can be employed to find the optimal solution.

\vspace*{2mm}
\section{%
Heuristic Approach via Circuit Decomposition
}
\label{sec:heuristic}

Above, we have seen how to reformulate Clifford synthesis as a SAT problem.
However, the search space for that problem grows exponentially with the maximal circuit depth $d_{\max}$ and the number of qubits $n$.
Depending on the specific SAT solver being used, this exact synthesis approach can quickly become prohibitively expensive. 
One way to diminish these scaling issues is to split a big Clifford circuit into a collection of sub-circuits that can be synthesized in parallel. This splitting can be done both horizontally (to reduce qubit number) as well as vertically (to reduce circuit depth)
and considerably reduces the size of the SAT search space. 
The result is a versatile heuristic (it cannot be guaranteed that the splitting into sub-circuits is optimal) that can be applied to larger Clifford circuits.

More precisely, let $G$ be a target Clifford circuit on $n$ qubits with maximum depth $d_{\max}$. Then, the associated SAT encoding features bitstrings of length $l=O \left(n^2 d_{\max}\right)$,
which corresponds to a search space of size $2^{l}=2^{O(n^2 d_{\max})}$.
We can now vertically split up $U=L_1 L_2$, where each $L_i$ has depth $d' \approx  d_{\max}/2 $, and apply our Clifford synthesis to each Clifford block. The result is two parallel SAT instances with bitstrings of length $l' \approx l/2$ each. In turn, the size of the search space is only $2^{l'}\approx\sqrt{2^l}$.
This quadratic improvement in search space size comes at the cost of a (potentially) non-optimal decomposition into two blocks. 

This general idea extends to more than two vertical blocks.
Let $G = L_1 \cdots L_m$ be a Clifford circuit with $m$ layers, i.e. blocks of qubits $L_i = g_i^1\cdots g_i^{k_i}$ such that all $g_i^{k_i}$ act on different qubits and can therefore be run in parallel.
Given a \emph{split size} $s$, we can partition the circuit into $\frac{m}{s}$ sub-circuits $G_i = L_{i\cdot s}\cdots L_{(i+1)\cdot s} $ for $0 \leq i < m$ (here it is assumed that $m$ is divisible by $s$ but extending the argument is straight-forward).
We then simply need to compute the target stabilizer tableau for each of these sub-circuits.
For $G_i$, this can be done by simulating the circuit up to $L_{i \cdot s}$.
We can then use the SAT reformulation proposed in \autoref{sec:sat} for each of the sub-circuits.
The split size $s$ can also provide a good initial guess as to the maximal depth needed for the encoding of each individual circuit.

Since no data needs to be shared between the individual instances, all the SAT instances can be run in parallel to obtain the optimized sub-circuits $G_i^\prime$ which are then concatenated for the final result, i.e., $G^\prime = G_1^\prime \cdots G_{\frac{m}{s}}^\prime$.
Note that such a splitting approach is guaranteed to produce a correct Clifford gate decomposition. This follows from applying \autoref{thm:main} to each of the subblocks involved. However, it may not achieve optimal circuit depth.
After all, this divide-and-conquer heuristic treats different blocks of the target circuit completely independently and scales with circuit depth and qubit number of the initial circuit.
Scaling issues can be countered by decreasing the split size, but the point of diminishing returns is eventually reached where the size of the split has a strong negative impact on the target metric.

\begin{table*}[t]
  \centering
  \caption{Experimental results for random Clifford circuits.}
  \label{tab:rand_cliff}
\vspace*{-3mm}
  \begin{tabular}[t]{c r r r r  r r r r r r r r}
    \toprule
    & \multicolumn{3}{c}{Optimal} & \multicolumn{3}{c}{Heuristic Vertical} & \multicolumn{3}{c}{Heuristic Horizontal} & \multicolumn{3}{c}{Bravyi et al.} \\ \cmidrule(lr){2-4} \cmidrule(lr){5-7} \cmidrule(lr){8-10} \cmidrule(lr){11-13}
    \multicolumn{1}{c}{$n$} & \multicolumn{1}{c}{$d$} &\multicolumn{1}{c}{$|G|$}& \multicolumn{1}{c}{$t~[\si{s}]$} & \multicolumn{1}{c}{$d$} &\multicolumn{1}{c}{$|G|$}&  \multicolumn{1}{c}{$t~[\si{s}]$} & \multicolumn{1}{c}{$d$}&\multicolumn{1}{c}{$|G|$} & \multicolumn{1}{c}{$t~[\si{s}]$} &\multicolumn{1}{c}{$d$}&\multicolumn{1}{c}{$|G|$} & \multicolumn{1}{c}{$t~[\si{s}]$} \\ 
    \begin{filecontents*}{csv/results.csv}
nqubits,dopt,gatesopt,topt,dvertical,gatesvertical,tvertical,dhorizontal,gateshorizontal,thorizontal,dgreedy,gatesgreedy,tgreedy
3,5.70,11.40,0.33,8.30,18.90,0.69,-,-,-,11.70,16.10,0.18
4,6.60,16.70,4.10,13.30,23.70,2.03,-,-,-,16.00,23.50,0.16
5,7.60,25.00,381.95,18.90,38.60,7.72,7.80,25.50,603.08,22.90,37.30,0.18
6,-,-,-,24.60,55.40,27.27,17.00,51.10,180.77,29.40,55.10,0.20
7,-,-,-,31.40,69.40,43.51,25.10,75.80,97.16,37.00,70.20,0.17
8,-,-,-,37.40,88.50,153.96,32.00,97.20,186.42,42.10,86.30,0.20
9,-,-,-,41.20,102.40,313.82,39.70,121.30,219.32,53.50,108.80,0.28
10,-,-,-,50.80,131.30,547.93,53.00,155.20,106.66,59.90,128.70,0.28
11,-,-,-,58.80,152.90,838.91,60.70,180.90,133.97,72.20,157.30,0.25
12,-,-,-,66.70,174.50,1464.46,75.80,217.40,117.73,78.90,170.10,0.26
13,-,-,-,75.50,206.50,2423.24,86.10,250.00,283.31,91.40,207.10,0.34
14,-,-,-,83.00,237.10,4412.97,99.80,305.70,111.89,100.30,235.20,0.33
\end{filecontents*}
\csvreader[head to column names]{csv/results.csv}{}{\\\nqubits&\dopt&\gatesopt&\topt&\dvertical&\gatesvertical&\tvertical&\dhorizontal&\gateshorizontal&\thorizontal&\dgreedy&\gatesgreedy&\tgreedy} \\\bottomrule
  \end{tabular}\\\vspace*{1mm}
  {\small $n$: Number of qubits \hspace*{0.20cm} d: Average depth \hspace*{0.20cm} $|G|$: Average gate count\hspace*{0.20cm} $t$: average runtime}\vspace*{2mm}
\vspace*{-3mm}
\end{table*}

Since any sub-circuit can be optimized using the SAT method without changing the circuit's functionality, we can take the divide-and-conquer approach even further.
Given a maximal number of qubits $n_\mathit{max}$, a circuit can be decomposed into sub-circuits $G = G_1 \cdots G_m$ such that the number of qubits in each circuit is bounded by $n_\mathit{max}$ and there are no two-qubit gates between any two sub-circuits for parallel optimization.

These two splitting techniques can be combined to make the approach as scalable as possible.
Given a depth threshold $d_\mathit{thr}$, a split size $s < d_\mathit{trh}$, and a maximal number of qubits $n_\mathit{max}$, the circuit can be split into sub-circuits of at most $n_\mathit{max}$ qubits.
If any of these sub-circuits is deeper than $d_\mathit{thr}$, the circuits can be further split horizontally into blocks of depth $s$. These circuit blocks can then be optimized independently from each other (in parallel). 

\section{Evaluations}
\label{sec:evaluations}

The methods proposed in \autoref{sec:sat} and \autoref{sec:heuristic} have been implemented in C++ using the publicly available SMT solver Z3~\cite{demouraZ3EfficientSMT2008}.
The implementation is integrated into the quantum circuit compilation tool QMAP~\cite{willeMQTQMAPEfficient2023}, which is part of the \emph{Munich Quantum Toolkit} (MQT) and available at \mbox{\url{https://github.com/cda-tum/qmap}}. 

To see how well the proposed methods perform in practice, we considered two types of benchmarks:
\begin{itemize}
\item[(i)]\emph{Random Clifford circuits} (inspired by randomized benchmarking~\cite{knillRandomizedBenchmarkingQuantum2008,magesanEfficientMeasurementQuantum2012,kuengComparingExperimentsFaulttolerance2016}): 
The circuits were obtained by sampling a random stabilizer tableau (including information about the destabilizers).
Since our SAT encoding is based on the stabilizer tableau, this is already a valid input format for our method and no explicit circuit has to be generated.
For every qubit number $n$, 10 random stabilizer tableaus have been generated and the results have been averaged over all these runs. The proposed methods are compared to the state of the art greedy Clifford synthesizer by Bravyi et al.~\cite{bravyiCliffordCircuitOptimization2021}
The timeout was set to $\SI{3}{\hour}$.

\item[(ii)] \emph{Clifford+T implementations of Grover search} (inspired by fault-tolerant quantum computation
~\cite{kitaevQuantumComputationsAlgorithms1997,groverFastQuantumMechanical1996,gottesmanTheoryFaulttolerantQuantum1998}): 
we generated circuits for the Grover search algorithm using random Boolean functions as oracles.
For each qubit number $n$, 10 circuits were generated in this fashion.
Since these circuits contain $T$-gates, the circuit is partitioned into Clifford blocks and $T$ blocks and each (non-trivial) Clifford block is optimized separately.
\end{itemize}

All evaluations have been performed on a a \SI{3.6}{\giga\hertz} Intel Xeon W-1370P machine running Ubuntu 20.04 with \SI{128}{\gibi\byte} of main memory and 16 hardware threads.
For generating and synthesizing circuits as well as for the greedy optimizer, the quantum computing SDK Qiskit~\cite{Qiskit-Textbook} (version 0.42.1) by IBM has been used.

The results of our experiments for synthesizing random Clifford circuits can be seen in \autoref{tab:rand_cliff}.

The data under column \emph{Optimal} shows the results using the proposed optimal SAT approach.
Unfortunately, the increase in variables and the scaling of the encoding in the number of qubits can be seen rather drastically here.
It is only possible to synthesize random Clifford circuits up to 5 qubits within the given time limit.
Nonetheless, we can see that other methods, especially the state of the art, synthesize $n$-qubit circuits that are far from depth-optimal, increasing the depth on average by $105.26\%$ ($n=3$), $142.42\%$ ($n=4$) and $201.32\%$ ($n=5$), respectively.

The column \emph{Heuristic Vertical} shows the results using the proposed heuristic approach where the circuits were partitioned vertically, i.e., the resulting sub-circuits all had the same number of qubits as the original circuit.
Compared to the vertical splitting heuristic, the state of the art still produces circuits that are $21.92\%$ deeper on average.

Column \emph{Heuristic Horizontal} shows the results using the proposed heuristic approach where the circuits were partitioned into sub-circuits of five qubits each. At this value, the optimal approach still yields results in an acceptable amount of time (thus there are no entries for three and four qubits).
This decomposition leads to much better results for lower qubit numbers but eventually produces worse results than the vertical decomposition after nine qubits -- partially as an artefact of the decomposition scheme.
The circuits cannot always be split perfectly into five qubit sub-circuits and potentially parallel gates in the original circuit might not be parallel anymore after the optimization.
Another reason is that it gets increasingly difficult to find deep sub-circuits of only five qubits for random Clifford circuits since the interactions between qubits are bound to entangle more than five qubits rather quickly.
A big upside of this synthesis method is its runtime. Since the runtime for synthesizing five qubit circuits is predictable, synthesizing all these sub-circuits can be done within a predictable time as well.

All in all, the results in \autoref{tab:rand_cliff} suggest that the heuristic approach could be improved with more sophisticated circuit decomposition techniques.
The current state of the art produces circuits that are far from the optimal, which leaves quite some room for improvement.

The Grover search benchmark was chosen to analyse the possible improvement of the depth of actual fault-tolerant quantum circuits.
Guided by the results of the random Clifford benchmarks we looked at Grover circuits for three to five qubits.
The proposed optimization scheme resulted in circuits that were $13.38\%$, $21.71\%$ and $16.72\%$ shorter on average.
The Grover benchmarks are made publicly available under \mbox{\url{https://github.com/cda-tum/qmap}}.

\section{Conclusion and Outlook}
\label{sec:conclusion}

Classical circuit synthesis, on the one hand, and Quantified Boolean Formulas (QBF), on the other, are two 
seemingly very different but related problems. This correspondence can be made one-to-one and forms the basis 
of several state of the art approaches for optimal (logical) circuit synthesis: QBF solvers are employed to determine the shortest circuit representation of a desired logical functionality. 

In this work, we have extended this general mindset to the quantum realm, considering the task of decomposing $n$-qubit Clifford circuits into as few elementary gates as possible. We then showed that deciding if is it possible to represent a given Clifford functionality with at most $d_{\max}$ Clifford layers can be re-cast as a satisfiability (SAT) problem in $O(n^2d_{\max})$ Boolean variables. The reduction uses maximally entangled input stimuli, as well as the Gottesman-Knill theorem. It highlights that Clifford synthesis contained in NP (the first level of the polynomial hierarchy) is easier than general logical circuit synthesis which is complete for $\Sigma_2^p$ (the second level of the polynomial hierarchy).

In the electronic design automation community, SAT solvers have been applied to tackle classical synthesis problems with great success. 
We showed that similar approaches apply to quantum computing and that there is large potential to replicate the success of solving classical synthesis problems.
While the optimal synthesis approach scales poorly in the number of qubits, it shows that there is a large gap between the optimal solution and the state of the art.
Up to our knowledge, these are the first Clifford synthesis protocols that (i) are provably correct and (ii) come with a certificate of optimality. 
Furthermore, all of the proposed methods are publicly available within the \emph{Munich Quantum Toolkit} (MQT) as part of the open-source quantum circuit compilation tool QMAP (\url{https://github.com/cda-tum/qmap}).
QMAP already works natively with IBM's Qiskit and even tighter integration with quantum SDKs is left for future work.

These initial findings are encouraging and open the door for several interesting follow-up projects, i.e. a \emph{refinement of the proposed numerical solver}. This will entail tweaks in the stabilizer encoding to squeeze out more performance, but also trying different SAT solvers and novel pre-processing techniques to determine sharper initial bounds on the maximum circuit depth. 
A \emph{configurable gate-set} is also on our to-do list. For now, we only use $H,S,S^\dagger, CNOT$, as well as Pauli gates. 
In the future, this could be adapted to include additional single-qubit gates (e.g.\ the full single-qubit Clifford group) and two-qubit gates (e.g. $SWAP$, $CZ$, $CY$, \ldots). Note that more `elementary' gates directly translate into a more complex encoding, and therefore a larger logical search space (more variables). On the other hand, this increased expressiveness per time step is bound to decrease the circuit depth and, therefore, result in shorter SAT formulas overall (fewer logical clauses). This trade-off might well be worthwhile.

The proposed encoding procedure is flexible enough to facilitate \emph{architecture-aware} synthesis of Clifford circuits. 
Some quantum architectures only support certain interactions between their qubits, typically defined by a coupling map.
We can respect this coupling map in the proposed SAT encoding by only permitting Clifford gates that are also native to the concrete architecture. 

Virtually all these future research directions also extend to improve the proposed \emph{heuristic solver} for near-optimal Clifford synthesis. 
We intend to explore different divide-and-conquer strategies (decomposing into sub-circuits) and explore (near-)optimal ways on how to best synthesize each of these circuit blocks.

Last but not least, researchers are now beginning to suggest and explore the use of quantum computers to solve challenging subroutines in quantum synthesis. Quantum assisted quantum compiling~\cite{khatriQuantumassistedQuantumCompiling2019} falls into this category. For this work, the deliberate restriction to Clifford circuits has allowed us to not have to think along these lines (yet) -- the Gottesman-Knill theorem ensures that classical simulation remains tractable throughout. But it is an interesting direction for future work to fruitfully combine quantum assisted quantum compiling ideas with conventional SAT-solving techniques. We leave such synergies for future work.

\section*{Acknowledgments}

The authors thank Armin Biere, David Gross, and Martina Seidl for inspiring discussions and valuable feedback. 

T.P., R.W.\  and L.B.\ acknowledge funding from the European Research Council (ERC) under the European Union’s Horizon 2020 research and innovation program (grant agreement No. 101001318), as well as financial support from the Munich Quantum Valley, which is supported by the Bavarian state government with funds from the Hightech Agenda Bayern Plus.
All authors have been supported by the BMWK on the basis of a decision by the German Bundestag through the projects ProvideQ and QuaST, the Project QuantumReady (FFG 896217) and the State of Upper Austria in the frame of the COMET program (managed by the FFG).

\printbibliography

\end{document}